\documentclass{article}

\usepackage{arxiv}

\usepackage[utf8]{inputenc} 
\usepackage[T1]{fontenc}    
\usepackage{hyperref}       
\usepackage{url}            
\usepackage{booktabs}       
\usepackage{amsmath, amssymb, amsfonts, amsthm}       
\usepackage{mathtools}
\usepackage{algorithm, algpseudocode}
\usepackage{algorithmicx}
\usepackage{nicefrac}       
\usepackage{microtype}      
\usepackage{graphicx}
\usepackage{bm}
\graphicspath{ {./img/} }

\newcommand{\sizeH}{\ensuremath{p\times n}}
\def\Tr{\mathrm{Tr}}
\def\BTr{\mathrm{BTr}}
\def\Exp{\mathrm{Exp}}
\def\R{\mathbb{R}}
\def\Pn{\mathbb{P}_n}
\def\Pp{\mathbb{P}_p}

\def\Sn{\mathbb{S}_n}
\def\half{{1/2}}
\def\Z{\mathcal{Z}}
\def\Y{\mathcal{Y}}
\def\N{\mathcal{N}}
\def\thhat{\hat{\theta}}

\def\grad{\mathrm{grad}}
\def\Hess{\mathrm{Hess}}
\def\sym{\mathrm{sym}}
\def\fbar{\bar{f}}
\def\Jbar{\bar{J}}
\def\D{\mathrm{D}}
\def\Pr{\mathrm{Pr}}
\def\phat{\hat{P}}
\def\popt{\prescript{o}{}{P}}
\def\khat{\hat{K}}
\def\ktendsto{\xrightarrow{k\to\infty}}
\def\ptendsto{\xrightarrow[\mathcal{P}]{k\to\infty}}
\newcommand{\vech}[1][\cdot]{\mathrm{vech}(#1)}
\newcommand{\uvec}[1][\cdot]{\mathrm{vec}(#1)}
\newcommand{\I}{\bm{I}}

\newtheorem{thm}{Theorem}

\newtheorem{lem}[thm]{Lemma}
\newtheorem{assum}[thm]{Assumption}
\newtheorem{prop}[thm]{Proposition}
\newtheorem{defn}[thm]{Definition}
\newtheorem{rem}[thm]{Remark}

\title{Riemannian Trust-Region based Adaptive Kalman filter with unknown noise Covariance matrices}

\author{
 Rahul Moghe, Maruthi R. Akella, and Renato Zanetti \\
  Aerospace Engineering and Engineering Mechanics\\
  The University of Texas at Austin\\
  Austin, TX 78703
}

\begin{document}
\maketitle
\begin{abstract}
The problem of adaptive Kalman filtering for a discrete observable linear
time-varying system with unknown noise covariance matrices is addressed in
this paper. The measurement difference autocovariance method is used to
formulate a linear least squares cost function containing the measurements
and the process and measurement noise covariance matrices. Subsequently, a
Riemannian trust-region optimization approach is designed to minimize the
least squares cost function and ensure symmetry and positive definiteness
for the estimates of the noise covariance matrices. The noise covariance
matrix estimates, under sufficient excitation of the system, are shown to
converge to their unknown true values. Saliently, the exponential stability
and convergence guarantees for the proposed adaptive Kalman filter to the
optimal Kalman filter with known noise covariance matrices is shown to be
achieved under the relatively mild assumptions of uniform observability and
uniform controllability.  Numerical simulations on a linear time-varying
system demonstrate the effectiveness of the proposed adaptive filtering
algorithm.
\end{abstract}


\section{Introduction}
\label{sec:introduction}

The Kalman filter for linear systems plays a foundational role
for modern statistical estimation theory~\cite{kalman1960new, kalman1961new}.
One of the important assumptions that allows for a reliable state estimation
algorithm is that the covariance matrices (CM) for the process and measurement
noise signals entering the system are perfectly known. This assumption is
rarely true in practice given the difficulty associated with obtaining perfect
models and characterization for noise parameters. Traditionally, the CMs are
either predetermined through extensive experimentation or are artificially
inflated to adopt a conservative strategy. The case when inaccurate noise
covariances are used is known to cause filter
divergence~\cite{sangsuk1990discrete, heffes1966effect, madjarov1996kalman,
saab1995discrete, sangsuk1988continuous, willems1992divergence}.  These
challenges motivate adaptive algorithms for state estimation while
simultaneously estimating the noise CMs.

Over the years, the noise CM estimation methods can be broadly classified into
Bayesian methods~\cite{lainiotis1971optimal, magnant2016bayesian,
huang2017novel}, maximum likelihood methods~\cite{shumway1982approach,
axelsson2011ml, kashyap1970maximum}, correlation
methods~\cite{mehra1970identification, belanger1974estimation,
dunik2016autocovariance, odelson2006new, aakesson2008generalized}, covariance
matching methods~\cite{sage1969adaptive, myers1976adaptive,
rajamani2009estimation}, subspace methods~\cite{mussot2021noise}, and predictor
error methods~\cite{lennart1999system} to name a few. A recent survey describes
and compares noise CM estimation methods~\cite{dunik2017noise}. Of all these
various existing solutions, the correlation methods have received significant
attention because they require the least a priori information, have less
computational requirements, are known to produce unbiased estimates, and in
some formulations are independent of the state estimates. For that reason,
correlation methods have also been termed as feedback-free in
Ref.~\cite{dunik2017noise}.

Among the various implementations of correlation methods, the autocovariance
least squares (ALS) approach evaluates the autocovariance of the innovation
sequence which is linearly dependent on the noise
CM~\cite{odelson2006new, aakesson2008generalized}. The Measurement
Autocovariance Method (MACM) method similarly builds a linear equation in the
noise CMs by evaluating the correlation of a modified measurement model formed
by stacking measurements in time~\cite{zhou1995estimation}. This method was
also shown to estimate the cross-covariance of the noises~\cite{lee1980direct}.
Recently, the noise CM estimates found using a measurement difference
autocovariance (MDA) approach were shown to converge to their true values in
the mean squares sense for linear time-varying (LTV) systems with mild
requirements of observability of the system~\cite{dunik2018design}. The noise
CM estimation technique developed in this paper introduces certain judiciously
determined modifications to the MDA approach to ensure a stable adaptive Kalman
filter (AKF) formulation.

In spite of notable advances in the estimation of noise CM, the adaptive
filtering problem that uses the noise CM estimates to estimate the states has
not been fully addressed. Our previous work derived a convergent AKF for
detectable linear time-invariant (LTI) systems~\cite{moghe2019adaptive}. The
convergence of the noise covariance estimates and the state error covariance
were established based on the full rank condition of the coefficients of the
noise CMs in the linear equation obtained from the MACM method. Although the
MDA method~\cite{dunik2018design} produced convergent noise CM estimates for
LTV systems, the stability of the AKF formulated using these noise CM estimates
was not analyzed. It needs to be noted that LTV systems represent a larger
class of systems that are ubiquitous in engineering applications. Additionally,
nonlinear systems with and extended Kalman filter formulation are also
represented as linear time-varying systems. As a result, an exponentially
stable AKF formulation for LTV systems with unknown noise CM has potential for
positively impacting many applications and is the topic of this paper.

The stability of the AKF is a non-trivial problem. In most correlation methods,
a least squares formulation estimates the noise CMs by vectorizing the linear
equation formed using the MACM, ALS, and MDA approaches. A recursive version of
the least squares problem allows for estimating the noise CMs on line. In this
setting, although the estimates of the noise CM are guaranteed to be symmetric
at all times, there are no assurances on their positive definiteness. In fact,
during the transients, the noise CM can sometimes violate positive definiteness
which can potentially lead to inconsistencies in the AKF. In order to
circumvent this issue, Ref.~\cite{moghe2019adaptive} adopts a convenient
remedy, that is to revert back the CM estimate to the most recent (prior) value
when it was symmetric and positive definite (SPD). In spite of the poor
transient performance, the overall convergence result still holds with this
heuristic in place as the estimates are proved to be SPD in the limit with
probability one. However, this motivates the question whether or not the
adaptive estimator can be modified to guarantee the SPD property for the CM
estimates at all times.

In this paper, we provide a positive answer to the foregoing question by
adopting a differential geometric approach that ensures SPD noise CM estimates.
Geometric optimization methods have gained popularity in the past decade
because firstly, the conformity of the optimized values to the geometry of the
set over which they are optimized is guaranteed and secondly, convergence to
the optimal value is shown to be faster as compared to their Euclidean
counterparts~\cite{absil2009optimization, boumal2014manopt}. Specifically, the
Riemannian trust-region optimization methods provide superior convergence
assurances compared to most other geometric optimization techniques that have
been used in various aspects of identification
theory~\cite{menegaz2018unscented, sato2019riemannian, zhang2018feedback}. The
set of SPD matrices form a Riemannian manifold when endowed with an appropriate
metric. Given a cost function, the Riemannian versions of the gradient and
Hessian enable geometric optimization techniques to be applied that ensure SPD
matrix estimates.

In this paper, a Riemannian trust-region (RTR) optimization based adaptive
Kalman filter is introduced that minimizes the recursive least squares cost
function and estimates the unknown noise CMs while simultaneously estimating
the states. Exponential stability of the adaptive Kalman filter has been
established under the uniform observability and uniform controllability assumptions
of the system regardless of whether or not the noise CM estimates converge. To
the best of our knowledge, such a stable formulation of an adaptive Kalman
filter has never been proved. Under the persistence of excitation condition,
the RTR based noise CM estimates are shown to converge in probability to its
true value. Under this condition, the state error covariance sequence of the
adaptive Kalman filter is also shown to converge to the classical non-adaptive Kalman
filter employing known values of the noise CMs.

The paper is organized as follows. Section~\ref{sec:akf} presents the adaptive
Kalman filter formulation using the MDA method. The RTR-based noise CM
estimation method is described in Section~\ref{sec:rtr}. Convergence of the
RTR-based noise CM estimates to their true values and the stability of the
adaptive Kalman filter formulated using the RTR-based noise CM estimates is
discussed in Section~\ref{sec:stab}.  Section~\ref{sec:sim} presents numerical
simulations that demonstrate the effectiveness of the RTR-based adaptive Kalman
filter (RTRAKF). Finally, the paper is concluded in Section~\ref{sec:conc} and
future work is discussed.

\noindent\textit{Notation:} The set of symmetric, and symmetric and positive
definite matrices are denoted by $\Sn$, and $\Pn$ respectively. The identity
matrix of size $n$ is denoted by $\I_n$. The operator
$\vech:\Sn\to\R^{n(n+1)/2}$ returns a vector with the unique elements of a
symmetric matrix and the operator $\uvec:\R^{n\times n}\to\R^{n^2}$ returns a
vector of all the elements of a matrix. The Kronecker product, denoted by
$\otimes$, operates on a matrix product as $\uvec[A X B] = (B^T\otimes A)
\uvec[X]$. A modified Kronecker product taking one sided symmetry into account
for $X\in\Sn$ is given by $\uvec[A X B] = (B^T\otimes_u A)\vech[X]$. Another
modified Kronecker product taking two-sided symmetry into account for $X\in\Sn$
is given by $\vech[A X A^T] = (A\otimes_h A)\vech[X]$. The operator $\Exp$
denotes the matrix exponential and $\Tr\{\cdot\}$ denotes the Trace operation
on a square matrix.  The Block Trace operation, denoted by $\BTr\{\cdot\}$ is
defined as $\Tr\{A(\I_n\otimes B)\} = \Tr\{\BTr\{A\} B\}$ wherein
$A\in\R^{n^2\times n^2}$ and $B\in\R^{n\times n}$. The symbol $\ktendsto$
represents the convergence as $k\to\infty$, and $\ptendsto$ denotes the
convergence in probability as $k\to\infty$. The notation $\Pr\{\cdot\}$ denotes
the probability of occurrence of an event.

\section{Adaptive Filter Formulation}\label{sec:akf}

The basic structure of the adaptive filter formulated in this section follows
from correlation based techniques~\cite{moghe2019adaptive, dunik2018design, lee1980direct}.

\subsection{Problem Formulation}\label{sub:prob}

A discrete linear time-varying (LTV) system is considered here with the system equations
given by
\begin{equation}\label{eq:sys}
	\begin{split}
		x_{k+1} &= F_k x_k + G_k u_k + w_k \\
		y_k &= H_k x_k + v_k
	\end{split}
\end{equation}
wherein the process noise $w_k\sim\N (\bm{0}_{n\times 1}, Q)$ and the measurement noise
$v_k\sim\N(\bm{0}_{p\times 1}, R)$ are uncorrelated white Gaussian noises with
constant noise covariance matrices. Let $\phi_{i,k}$ be the associated state
transition matrix such that $\phi_{k+1,k} = F_k$ and $\phi_{k,l} = \phi_{k,q}\phi_{q,l}$ for any $q\in Z$. No restrictions are made for
the matrices $F_k$, $G_k$ and $H_k$ except for uniform observability of pair
$(F_k,H_k)$ and uniform controllability of the $(F_k, Q_k^\half)$ pair. To that
end, the definition for uniform observability and uniform controllability is
given by~\cite[Chapter 7.5]{jazwinski2007stochastic}
\begin{defn}\label{defn:obs}
	The pair $(F_k, H_k)$ is uniformly observable if there exists an integer $s
	\geq 0$ and constants $0 < \alpha_1 < \alpha_2$ such that
	\begin{equation}\label{eq:obsGrampos}
		\alpha_1 \I \preceq M_{k+s,k} \preceq \alpha_2 \I
	\end{equation}
	wherein,
	\begin{equation}\label{eq:obsgram}
		M_{k+s,k} = \sum\limits_{i=k}^{k+s} \phi_{i,k}^T H_i^TH_i\phi_{i,k}
	\end{equation}
	and let $M_{k+l,k}$ for $l<s$ be the partial observability Gramian.
\end{defn}
\begin{defn}\label{defn:con}
	The pair $(F_k, E_k)$ is uniformly controllable if there exists an integer
	$s\geq 0$ and constants $0 < \beta_1 < \beta_2$ such that
	\begin{equation}\label{eq:conGrampos}
		\beta_1 \I \preceq Y_{k+s,k} \preceq \beta_2 \I
	\end{equation}
	wherein,
	\begin{equation}\label{eq:congram}
		Y_{k+s,k} = \sum\limits_{i=k}^{k+s} \phi_{k+s+1,i+1} E_i E_i^T \phi_{k+s+1,i+1}^T
	\end{equation}
\end{defn}
\begin{assum}\label{as:obscon}
Let the pair $(F_k, H_k)$ be uniformly observable and the pair $(F_k,
Q_k^{\half})$ be uniformly controllable.
\end{assum}
The above assumption ensures the
Kalman filter is exponentially stable~\cite[Theorem
5.3]{anderson1981detectability}. Subsequently, the following assumption is made
on the $Q$ and $R$ noise covariance matrices.
\begin{assum}\label{as:covunknown}
	The noise covariance matrices $Q$ and $R$ are both assumed to be constant and unknown.
\end{assum}
The aim of this paper is to estimate the unknown $Q$ and $R$ matrices while
simultaneously estimating the states.

\subsection{Measurement Difference Autocovariance approach}\label{sub:mda}

Since the pair $(F_k,H_k)$ is uniformly controllable with constants $s\geq 0$
and $0 < \alpha_1 < \alpha_2$, consider $m\geq s$ measurements that are
aggregated in time to form a linear time series. Such a development is
described in the following result.
\begin{prop}
	For the LTV system given by Eq.~\eqref{eq:sys} with the
	Assumptions~\ref{as:covunknown}, and a $m\geq s$ from
	Definition~\ref{defn:obs}, the measurements $y_k$ of the system follow a
	linear time series given by
	\begin{equation}\label{eq:timeseries}
		\sum\limits_{i=0}^{m} A_i^k y_{k-i} - \sum\limits_{i=1}^{m} B_i^k G_{k-i} u_{k-i} = \sum\limits_{i=1}^{m} B_i^k w_{k-i} + \sum\limits_{i=0}^{m} A_i^k v_{k-i}
	\end{equation}
	wherein, the coefficients $A_i^k$ and $B_i^k$ are completely determined
	from the system matrices $F_k$ and $H_k$.
\end{prop}
\begin{proof}
	The proof follows from much of the past work on the measurement
	difference methods~\cite{dunik2018design, moghe2019adaptive}. We begin by
	accumulating $m$ measurements by stacking them one on top of the other to
	form a modified measurement model given in Eq.~\eqref{eq:stackmeas}.
	\begin{table*}[!htbp]
	\begin{equation}\label{eq:stackmeas}
		\begin{split}
		&\underbrace{\begin{bmatrix}
				y_k \\ y_{k-1} \\ \vdots \\ y_{k-m+1}
		\end{bmatrix}}_{\triangleq \mathcal{Y}_{k,k-m+1}} =
		\underbrace{\begin{bmatrix}
				H_k\phi_{k,k-m+1} \\ H_{k-1}\phi_{k-1,k-m+1} \\ \vdots \\ H_{k-m+1}
		\end{bmatrix}}_{\triangleq O_{k,k-m+1}} x_{k-m+1} \\
		&+ \underbrace{\begin{bmatrix}
				H_k & H_kF_{k-1} & H_k\phi_{k,k-2} & \cdots & H_k\phi_{k,k-m+2} \\
				\bm{0}_{\sizeH} & H_{k-1} & H_{k-1}F_{k-2} & \cdots & H_{k-1}\phi_{k-1,k-m+2} \\
			\vdots & \vdots & \ddots & \vdots & \vdots \\
			\bm{0}_{\sizeH} & \bm{0}_{\sizeH}& \bm{0}_{\sizeH}& \bm{0}_{\sizeH}& H_{k-m+2} \\
			\bm{0}_{\sizeH} & \bm{0}_{\sizeH}& \bm{0}_{\sizeH}& \bm{0}_{\sizeH}& \bm{0}_{\sizeH}
		\end{bmatrix}}_{\triangleq M^w_{k-1,k-m+1}}
		\underbrace{\begin{bmatrix}
				w_{k-1} + G_{k-1} u_{k-1} \\ w_{k-2} + G_{k-2} u_{k-2} \\ \vdots \\ w_{k-m+1} + G_{k-m+1} u_{k-m+1}
		\end{bmatrix}}_{\triangleq U_{k-1,k-m+1}} +
		\underbrace{\begin{bmatrix}
			v_k \\ v_{k-1} \\ \vdots \\ v_{k-m+1}
		\end{bmatrix}}_{\triangleq V_{k,k-m+1}}
		\end{split}
	\end{equation}
	\end{table*}
	Defining $W_{k-1,k-m+1} = [w_{k-1}^T,\ldots,w_{k-m+1}^T]^T$ and pre
	multiplying by $O_{k,k-m+1}$, the invertible observability Gramian
	$M_{k,k-m+1}$ defined in Eq.~\eqref{eq:obsgram} is recovered as shown
	below.
	\begin{align*}
		&O_{k,k-m+1}^T \Y_{k,k-m+1} = M_{k,k-m+1} x_{k-m+1} + \\
		&O_{k,k-m+1}^TM^w_{k-1,k-m+1}U_{k-1,k-m+1} + O_{k,k-m+1}^T V_{k,k-m+1}
	\end{align*}
	Inverting $M_{k,k-m+1}$ and using a one step predictor for the state
	$x_{k-m+1}$, a linear time series can be formed from the two equations
	given by
	\begin{align*}
		M_{k,k-m+1}^{-1} O_{k,k-m+1}^T \Y_{k,k-m+1} = x_{k-m+1} + M_{k,k-m+1}^{-1} &O_{k,k-m+1}^TM^w_{k-1,k-m+1} U_{k-1,k-m+1} \\
																				   &+ M_{k,k-m+1}^{-1} O_{k,k-m+1}^T V_{k,k-m+1} \\
		M_{k-1,k-m}^{-1} O_{k-1,k-m}^T \Y_{k-1,k-m} = x_{k-m} + M_{k-1,k-m}^{-1} &O_{k-1,k-m}^TM^w_{k-2,k-m} U_{k-2,k-m} \\
															  &+ M_{k-1,k-m}^{-1} O_{k-1,k-m}^T V_{k-1,k-m}
	\end{align*}
	Substituting $x_{k-m+1} = F_{k-m}x_{k-m} + G_{k-m}u_{k-m} + w_{k-m}$ and
	eliminating the state by subtraction, we get a linear time series given by
	\begin{equation}\label{eq:ltvseries}
		\mathcal{A}_k\Y_{k,k-m} = \mathcal{B}_k U_{k-1,k-m} + \mathcal{A}_k V_{k,k-m}
	\end{equation}
	wherein,
	\begin{equation}
		\mathcal{A}_k = [M_{k,k-m+1}^{-1} O_{k,k-m+1}^T, \bm{0}_{n\times p}] - [\bm{0}_{n\times p}, F_{k-m}M_{k-1,k-m}^{-1} O_{k-1,k-m}^T]
	\end{equation}
	and
	\begin{equation}
		\mathcal{B}_k = [M_{k,k-m+1}^{-1} O_{k,k-m+1}^TM^w_{k-1,k-m+1}, \I_{n\times n}] - [\bm{0}_{n\times n}, F_{k-m}M_{k-1,k-m}^{-1} O_{k-1,k-m}^TM^w_{k-2,k-m}]
	\end{equation}
	Separating out individual components of the coefficients of $y_k$
	\begin{align*}
		A_0 &= M_{k,k-m+1}^{-1}\phi_{k,k-m+1}^T H_k^T \\
		A_i &= M_{k,k-m+1}^{-1}\phi_{k-i,k-m+1}^T H_{k-i}^T - F_{k-m}M_{k-1,k-m}^{-1}F_{k-m}^T\phi_{k-i,k-m+1}^T H_{k-i}^T \\
		A_m &= -F_{k-m}M_{k-1,k-m}^{-1}F_{k-m}^T H_{k-m}^T
	\end{align*}
	wherein, $i = 1,\ldots,m-1$ above and the coefficients of $w_k$ is given by
	\begin{align*}
		B_1 &= M_{k,k-m+1}^{-1}\phi_{k,k-m+1}^T M_{k,k} \\
		B_i &= M_{k,k-m+1}^{-1} \phi_{k-i+1,k-m+1}^T M_{k,k-i+1} - F_{k-m} M_{k-1,k-m}^{-1} \phi_{k-i+1,k-m}^T M_{k-1,k-i+1} \\
		B_m &= \I_{n\times n} - F_{k-m}M_{k-1,k-m}^{-1}F_{k-m}^T M_{k-1,k-m+1}
	\end{align*}
	wherein, $i = 2,\ldots,m-1$ above. The statement of the proposition
	follows.
\end{proof}
Defining $\Z_k$ as the left hand side of Eq.~\eqref{eq:timeseries}, the
autocovariance function of $\Z_k$ is given by
\begin{equation}\label{eq:autocov}
	C_{k,k-p} = E[\Z_k \Z_{k-p}^T] = \sum\limits_{i=p+1}^m B^k_i Q {B^{k-p}_{i-p}}^T + \sum\limits_{i=p}^m A^k_i R {A^{k-p}_{i-p}}^T
\end{equation}
wherein $p=0,\ldots,m$. Notice that the autocovariance $C_{k,k-p} =
\bm{0}_{n\times n}$ for $p>m$ vanishes. As long as the number of past
measurements $y_k$ stored at every time instant is greater than $m+1$, the
autocovariance function can be estimated.

\subsection{Covariance Matrix Estimation}\label{sub:covest}

The autocovariance is estimated using a single measurement as $\hat{C}_{k,k-p}
= \Z_k \Z_{k-p}^T$. The elements of the autocovariance function can be
rearranged using the $\vech$ operation as follows.
\begin{equation}\label{eq:lincoveqn}
	\underbrace{\begin{bmatrix}
		\vech[\hat{C}_{k,k}] \\ \uvec[\hat{C}_{k,k-1}] \\ \vdots \\ \uvec[\hat{C}_{k,k-P}]
	\end{bmatrix}}_{\triangleq\, b_k}
	= \underbrace{\begin{bmatrix}
	\sum\limits_{i=1}^m B^k_i\otimes_h B^k_i & \sum\limits_{i=0}^m A^k_i\otimes_h A^k_i \\
	\sum\limits_{i=2}^m B^{k-1}_{i-1}\otimes_u B^k_i & \sum\limits_{i=1}^m A^{k-1}_{i-1}\otimes_u A^k_i \\
	\vdots & \vdots \\
	\sum\limits_{i=P+1}^m B^{k-P}_{i-P}\otimes_u B^k_i & \sum\limits_{i=P}^m A^{k-P}_{i-P}\otimes_u A^k_i
	\end{bmatrix}}_{\triangleq\, D_k}
	\underbrace{\begin{bmatrix}
		\vech[Q] \\ \vech[R]
	\end{bmatrix}}_{\triangleq\, \theta}
\end{equation}
A recursive least squares (RLS) estimation technique starting from an initial
guess $(\thhat_0, \Psi_0)$ is given by
\begin{equation}\label{eq:recursivels}
	\begin{split}
		\thhat_{k+1} &= \thhat_k + L_k(b_k - D_k\thhat_k) \\
		\Psi_{k+1} &= (\I - L_k D_k)\Psi_k (\I - L_k D_k)^T + L_k R_W L_k^T \\
		L_k &= \Psi_k D_k^T(R_W + D_k \Psi_k D_k^T)^{-1}
	\end{split}
\end{equation}
The convergence of the estimate has been established in Ref.~\cite[Theorem
8]{dunik2018design}.

\subsection{Adaptive Kalman Filter}\label{sub:akf}
Using the estimates $\hat{Q}_k$ and $\hat{R}_k$ of the noise covariance
matrices, the following equations constitute the adaptive Kalman filter
equations.
\begin{equation}\label{eq:akf}
	\begin{split}
		\hat{x}_{k|k-1} &= F_{k-1} \hat{x}_{k-1} + G_{k-1} u_{k-1} \\
		\hat{x}_{k|k} &= \hat{x}_{k|k-1} + \khat_k (y_k - H_k \hat{x}_{k|k-1})  \\
		\phat_{k|k-1} &= F_{k-1}\phat_{k-1|k-1}F_{k-1}^T + \hat{Q}_k  \\
		\khat_k &= \phat_{k|k-1}H_k^T (H_k\phat_{k|k-1}H_k^T + \hat{R}_k)^{-1}  \\
		\phat_{k|k} &= (\bm{I}- \khat_k H_k)\phat_{k|k-1}(\bm{I}- \khat_k H_k)^T + \khat_k \hat{R}_k \khat_k^T
	\end{split}
\end{equation}
wherein, $\phat_{k|k}$, $\phat_{k|k-1}$, and $\khat_k$ are the estimates of the
quantities in the nominal Kalman filter~\cite[Chapter 7]{jazwinski2007stochastic}.

\section{Riemannian Trust-Region method}\label{sec:rtr}

Although the recursive least squares successfully estimates the elements of the
noise covariance matrices, it does not guarantee SPD estimates of the
covariance matrix. The convergence of the estimates to the true covariance
matrices is guaranteed provided the matrix $D_k$ is persistently excited.
However, the transients are important when the covariance estimate is
concurrently used to estimate the state vector.  In this case, the filter may
run into a problem of loss of observability or worse, provide negative
information updates to the filter by virtue of a non positive definite noise
covariance matrix estimate. As a result, having a SPD noise covariance matrix
estimate is crucial to obtain a stable adaptive Kalman filter. To this end, a
geometric optimization approach that respects the geometry of SPD matrices is
introduced here. A brief summary of the geometry of SPD matrices is provided
below (for a comprehensive review, see, e.g.,~\cite{bhatia2009positive} for SPD
matrices and~\cite{absil2009optimization} for Riemannian optimization methods).

\subsection{Geometry of Covariance Matrices}\label{sub:geom}
The space $\Pn$ forms a manifold with its tangent space at a point $X\in\Pn$
denoted by $T_X \Pn$ and identified with $\Sn$, the set of symmetric matrices.
The affine invariant metric at $X\in\Pn$ defined by
\begin{equation}\label{eq:aimetric}
	\langle V_1, V_2 \rangle_X = \Tr\{X^{-1} V_1 X^{-1} V_2\}\quad V_1,V_2\in T_X \Pn
\end{equation}
turns the manifold into a Riemannian manifold. The shortest path on the manifold between
two points $X,Y\in\Pn$ is called the geodesic curve and is parameterized as
\begin{equation}\label{eq:geod}
	\gamma(s) = X^\half \left (X^{-\half} Y X^{-\half}\right)^s X^\half \quad s\in [0,1]
\end{equation}
wherein, $\gamma(0)=X$ and $\gamma(1) = Y$ denote the end points of the
geodesic. A geodesic curve emanating from a point $X\in\Pn$ in the direction
$V\in T_X\Pn$ is parameterized by
\begin{equation}\label{eq:geodtangent}
	\gamma_{X,V}(s) = X^\half \Exp \left(s X^{-\half} V X^{-\half}\right) X^\half
\end{equation}
and resides within $\Pn$ for any $s\in\R$. Given a smooth
function $f:\Pn\to\R$, $\fbar$ as the extension of $f$ to $\R^{n\times n}$, a
smooth geodesic curve $\gamma:\R\to\Pn$ such that $\gamma(0)=X\in\Pn$ and
$\dot{\gamma}(0)=V\in T_X\Pn$, the Euclidean gradient $\nabla \fbar$ defined
using the directional derivative $\D\fbar(X)[V]$ of $\fbar$ at $X$ in the
direction $V$ is given as
\begin{equation}\label{eq:eucgrad}
	\Tr\{V\, \nabla\fbar(X)\} = \D\fbar (X) [V]
\end{equation}
The Riemannian gradient of $f$ at $X$, denoted by $\grad\, f (X) \in T_X\Pn$ is
similarly defined as
\begin{equation}\label{eq:riemgrad}
	\langle V, \grad f(X)\rangle_X = \frac{d}{dt} f(\gamma(t))\big|_{t=0}
\end{equation}
Note that the Riemannian gradient is obtained from the Euclidean gradient by
\begin{equation}\label{eq:euctoriemgrad}
	\grad f(X) = X\sym(\nabla\fbar(X)) X
\end{equation}
From~\cite[Section 4.1.4]{jeuris2012survey}, the Riemannian Hessian of $f$
defined as a map $\Hess f(X):T_X \Pn\to T_X\Pn$ is given by
\begin{equation}\label{eq:riemhess}
	\Hess f(X)[V] = \D (\grad f)(X)[V] - \sym (\grad f(X)X^{-1} V)
\end{equation}
Using the above expressions for $\grad f$, the Hessian can be expressed in
terms of the extension $\fbar$ as
\begin{equation}\label{eq:euctoriemhess}
	\Hess f(X)[V] = X\sym(\D (\nabla\fbar)(X)[V])X + \sym(V\sym(\nabla\fbar) X)
\end{equation}

\subsection{Cost function, Gradient and Hessian}\label{sub:costgradhess}
The cost function for the recursive least squares minimization from
Eq.~\eqref{eq:recursivels} is minimized with a Riemannian optimization
framework. The recursive least squares cost function is given by
\begin{equation}\label{eq:lscost}
		J_k(\theta) = \frac{1}{2}(D_k\theta - b_k)^T R_W^{-1} (D_k\theta - b_k) + \frac{1}{2}(\theta - \thhat_{k-1})^T \Psi^{-1}_{k-1} (\theta - \thhat_{k-1})
\end{equation}
Before evaluating the Riemannian gradient and the Riemannian Hessian, the cost
function must be reformatted to explicitly depend on $Q$ and $R$ matrices. Such
reformatting is possible via simple algebraic manipulation. The unique elements
of a SPD matrix are given by
\begin{align}
	\vech[X] = \begin{bmatrix} \I^0_n X e_i \\ \vdots \\ \I^{n-1}_n X e_n \end{bmatrix} =
	\underbrace{\begin{bmatrix} \I^0_n & \bm{0}_{} & \cdots & \bm{0}_{} \\ \bm{0}_{} & \I^1_n & \cdots & \bm{0} \\
	\vdots & \vdots & \ddots & \vdots \\
	\bm{0}_{} & \bm{0} & \cdots & \I^{n-1}_n
	\end{bmatrix}}_{\triangleq \mathcal{I}_n\in\R^{n(n+1)/2 \times n^2}}\,
	(\I_n \otimes X) \uvec[\I_n]
\end{align}
wherein, $X\in\Sn$, $e_i\in\R^n$ is the $i^{th}$ canonical basis vector and
$\I^i_n\in\R^{(n-i)\times n}$ is formed by deleting the first
$i$ rows of $\I_n$, the identity matrix. The following statement provides the
expressions for the gradient of the least squares cost function.
\begin{lem}\label{lm:grad}
	Given the cost function in Eq.~\eqref{eq:lscost}, its Riemannian gradients at $Q$ and $R$ are given by
	\begin{equation}\label{eq:grad}
		\grad J_k (Q,R) = \left(Q\, \nabla_Q\Jbar_k\, Q, R\, \nabla_R\Jbar_k\, R\right)
	\end{equation}
	wherein, $\Jbar_k$ is the Euclidean extension of the cost $J_k$ and the
	expressions for $\nabla_Q \Jbar_k$ and $\nabla_Q\Jbar_k$ are given by
	\begin{equation}\label{eq:gradQ}
			\nabla_Q \Jbar_k = \BTr_n \left\{\sym\left(\uvec[\I_n] \left(D_k^{Q^T} R_W^{-1} (D_k\theta - b_k) + [\I_{m_q}, \bm{0}_{m_q\times m_r}]\Psi^{-1}_{k-1} (\theta - \thhat_{k-1}) \right)^T \mathcal{I}_n\right)\right\}
	\end{equation}
	\begin{equation}\label{eq:gradR}
		\nabla_R \Jbar_k = \BTr_p \left\{\sym\left(\uvec[\I_p] \left(D_k^{R^T} R_W^{-1} (D_k\theta - b_k) + [\bm{0}_{m_r\times m_q}, \I_{m_r}]\Psi^{-1}_{k-1} (\theta - \thhat_{k-1}) \right)^T \mathcal{I}_p\right)\right\}
	\end{equation}
	wherein, $D_k = [D_k^Q, D_k^R]$, $\theta = [\vech[Q]^T, \vech[R]^T]^T$, $m_q = n(n+1)/2$ and
	$m_r=p(p+1)/2$.
\end{lem}
\begin{proof}
  Consider a geodesic $\gamma_{Q,V_Q}(t)$ as defined in
  Eq.~\eqref{eq:geodtangent}. From the definition of the gradient in
  Eq.~\eqref{eq:eucgrad}, the expression for $\nabla_Q\Jbar_k$ is given by
  \begin{align*}
	  \D \Jbar_k(Q,R)[V_Q, V_R] = \frac{\partial \Jbar_k(\gamma_{Q,V_Q}(s), \gamma_{R,V_R}(s))}{\partial s}\Big|_{s=0} = \Tr\{\nabla_Q\Jbar_k\, V_Q\} + \Tr\{\nabla_R\Jbar_k\, V_R\}
  \end{align*}
  Separating the expressions into the parts containing $Q$ and $R$, we get
  \begin{align*}
  \Tr\{\nabla_Q\Jbar_k\, V_Q\}&= (D_k\theta - b_k)^T R_W^{-1} \left(D_k^Q\mathcal{I}_n \I_n\otimes V_Q \uvec[\I_n] \right) + (\theta - \thhat_{k-1})^T \Psi_{k-1}^{-1}[\I_{m_q}, \bm{0}_{m_q\times m_r}]^T \mathcal{I}_n \I_n\otimes V_Q \uvec[\I_n] \\
							  &= \Tr\{\uvec[\I_n] \left( D_k^{Q^T} R_W^{-1} (D_k \theta - b_k) + [\I_{m_q}, \bm{0}_{m_q\times m_r}]\Psi_{k-1}^{-1}(\theta - \thhat_{k-1}) \right)^T \mathcal{I}_n \I_n\otimes V_Q\}
  \end{align*}
  Further simplification results in an expression given by
  \begin{align*}
	  \Tr\{\nabla_Q\Jbar_k\, V_Q\}   = \Tr\left\{\BTr_n\left\{\uvec[\I_n] \left(D_k^{Q^T} R_W^{-1} (D_k \theta - b_k) + [\I_{m_q}, \bm{0}_{m_q\times m_r}]\Psi_{k-1}^{-1}(\theta - \thhat_{k-1})\right)^T \mathcal{I}_n\right\} V_Q\right\}
  \end{align*}
  wherein, the expression $\Tr\{A (\bm{I}_n\otimes B)\} = \Tr\{\BTr_n\{A\} B\}$ is
  used. Comparing the expressions on both sides of the equations gives the
  result of the lemma. The expression for $\nabla_R\Jbar_k$ results from a
  derivation similar to the one above and is omitted.
\end{proof}

The expression for the Riemannian Hessian of $J_k$ can be derived from
Eq.~\eqref{eq:euctoriemhess} and is given through the following statement.
\begin{lem}\label{lm:hess}
	The expression for the Riemannian Hessian of $J_k$ from
	Eq.~\eqref{eq:euctoriemhess} is given by
	\begin{equation}\label{eq:hess}
		\Hess J_k(Q, R) [V_Q, V_R] = \left( \Hess_Q J_k(Q,R)[V_Q, V_R], \Hess_R J_k(Q,R)[V_Q, V_R]\right)
	\end{equation}
	wherein
	\begin{align*}
		\Hess_Q J_k (Q,R)[V_Q, V_R] &= Q\sym(\D (\nabla_Q\Jbar_k)(Q,R)[V_Q, V_R])Q + \sym(V_Q\sym(\nabla_Q\Jbar_k) Q) \\
		\Hess_R J_k (Q,R)[V_Q, V_R] &= R\sym(\D (\nabla_R\Jbar_k)(Q,R)[V_Q, V_R])R + \sym(V_R\sym(\nabla_R\Jbar_k) R)
	\end{align*}
	The expressions for the directional derivatives of the Euclidean gradients are given by
	\begin{equation}\label{eq:dirderivQ}
			\D (\nabla_Q \Jbar_k) (Q,R)[V_Q, V_R] = \BTr_n \left\{\sym\left(\uvec[\I_n]  \theta_v^T \left(D_k^{Q^T} R_W^{-1} D_k + [\I_{m_q}, \bm{0}_{m_q\times m_r}]\Psi^{-1}_{k-1}\right)^T \mathcal{I}_n\right)\right\}
	\end{equation}
	and
	\begin{equation}\label{eq:dirderivR}
			\D (\nabla_R \Jbar_k) (Q,R)[V_Q, V_R] = \BTr_p \left\{\sym\left(\uvec[\I_p] \theta_v^T \left(D_k^{R^T} R_W^{-1} D_k + [\bm{0}_{m_r\times m_q}, \I_{m_r}]\Psi^{-1}_{k-1}\right)^T \mathcal{I}_p\right)\right\}
	\end{equation}
	wherein, $\theta_v = [\vech[V_Q]^T,\, \vech[V_R]^T]^T$.
\end{lem}
\begin{proof}
	The directional derivative of $\nabla_Q \Jbar_k (Q,R)[V_Q, V_R]$ is obtained as
	\begin{align}
		D (\nabla_Q \Jbar_k) (Q, R) [V_Q, V_R] = \frac{\partial \nabla_Q\Jbar_k (\gamma_{Q, V_Q}(s), \gamma_{R,V_R} (s))}{\partial s}\Big|_{s=0}
	\end{align}
	Since the gradient is affine in $Q$ and $R$, the directional derivative is
	independent of the points $Q$ and $R$ where it is evaluated. The expression
	is obtained trivially by substituting $V_Q$ and $V_R$ in place of $Q$ and
	$R$ and removing the constant terms.
\end{proof}

\subsection{Riemannian Trust-Region Method}\label{sub:rtr}

The Riemannian trust-region (RTR) method is used to solve the quadratic least squares
cost function in Eq.~\eqref{eq:lscost}. At each step the RTR method performs an
inner iteration that minimizes a quadratic approximation of a cost function at
$Q, R$ given by
\begin{equation}\label{eq:innerquadmodel}
		\hat{m}_{(Q,R)}(V_Q, V_R) = J_k (Q, R) + \langle \grad J_k(Q,R), (V_Q, V_R) \rangle_{(Q,R)} + \frac{1}{2} \langle \Hess J_k (Q,R)[V_Q, V_R], (V_Q, V_R) \rangle_{(Q,R)}
\end{equation}
The optimal $V_Q^*\in T_Q\Pn$ and $V_R^*\in T_R\Pp$ are obtained subject to a
norm constraint on step size given by
\begin{align}\label{eq:normconstraint}
	\|(V_Q, V_R) \|_{(Q,R)} = \sqrt{\langle (V_Q,V_R), (V_Q, V_R) \rangle_{(Q,R)}} \leq \Delta
\end{align}
wherein, $\Delta$ is the trust-region radius. A truncated conjugate gradient
(tCG) method~\cite[Algorithm 11]{absil2009optimization} solves the inner
iteration at each step. Then a verification step evaluates the decrease in the
true and approximate cost function given by the ratio
\begin{equation}\label{eq:ratio}
	\rho = \frac{J_k(Q,R) - J_k(\gamma_{Q,V_Q^*}(1), \gamma_{R,V_R^*}(1))}{\hat{m}_{(Q,R)}(\bm{0}_{n\times n}, \bm{0}_{p\times p}) - \hat{m}_{(Q,R)}(V_Q^*, V_R^*)}
\end{equation}
and decide whether the optimal $(V_Q^*, V_R^*)$ are accepted and whether the radius
$\Delta$ should be decreased. Algorithm~\ref{alg:rtr} describes the RTR
algorithm. The constants used are taken from~\cite[Algorithm
10]{absil2009optimization}.

\begin{algorithm}
\caption{Riemannian Trust-Region Method}\label{alg:rtr}
\begin{algorithmic}[0]
\State \textbf{Input:} $\hat{Q}_{k-1}$, $\hat{R}_{k-1}$, $\Psi_{k-1}$, $D_k$, $b_k$, $R_W$, $\bar{\Delta} > 0$, $\Delta_1 \in (0,\bar{\Delta})$, $\rho\in[0,\frac{1}{4})$
\State \textbf{Initialization:} $(\hat{Q}_k)_1 = \hat{Q}_{k-1}$, $(\hat{R}_k)_1=\hat{R}_{k-1}$
\State \textbf{Output:}  $\hat{Q}_k$, $\hat{R}_k$
\For{$k = 1 \to n$}
\State Minimize $\hat{m}_{(Q,R)}(V_Q,V_R)$  \Comment{Eq.~\eqref{eq:innerquadmodel}}
\State subject to norm constraint with $\Delta_i$ \Comment{Eq.~\eqref{eq:normconstraint}}
\If{$\rho_i<\frac{1}{4}$}
	\State $\Delta_{i+1} = \frac{1}{4}\Delta_i$
\ElsIf{$\rho_i > \frac{3}{4}$ and $\|(V_Q^*)_i, (V_R^*)_i\|=\Delta_i$}
	\State $\Delta_{i+1} = \min(2\Delta_i,\bar{\Delta})$
\Else
	\State $\Delta_{i+1} = \Delta_i$
\EndIf
\If{$\rho_i > \rho_{min}$}
	\State $((\hat{Q}_k)_{i+1}, (\hat{R}_k)_{i+1}) = (\gamma_{Q,V_Q^*}(1), \gamma_{r,V_R^*}(1))$
\Else
	\State $((\hat{Q}_k)_{i+1}, (\hat{R}_k)_{i+1}) = ((\hat{Q}_k)_i, (\hat{R}_k)_i)$
\EndIf
\EndFor
\end{algorithmic}
\end{algorithm}

\begin{rem}\label{rm:lmin}
	The RTR algorithm ensures that the estimates are symmetric and positive
	definite. However, in practice, the SPD noise covariance estimates may be
	arbitrarily close to semidefiniteness. This may create numerical errors in
	the filter updates. To avoid this situation, the minimum eigenvalue of
	$\hat{Q}_k$ and $\hat{R}_k$ is lower bounded by a small positive constant
	$\epsilon > 0$. The modified optimization variable is given by
	\begin{equation}\label{eq:modqr}
		\begin{split}
			\hat{Q}_k &= \epsilon\I_n + \hat{Q}^{\epsilon}_k \\
			\hat{R}_k &= \epsilon\I_p + \hat{R}^{\epsilon}_k
		\end{split}
	\end{equation}
	Such a modification ensures that the eigenvalues of the noise covariance estimates obtained by
	the RTR method are lower bounded by $\epsilon$ instead of zero. Such a
	modification merely results in a shift of the origin and does not affect
	the RLS solution.
\end{rem}

The algorithm for the RTR-based AKF is summarized below.

\begin{algorithm}
\caption{Riemannian Trust-Region based Adaptive Kalman Filter (RTRAKF)}\label{alg:RTRAKF}
\begin{algorithmic}[0]
\State \textbf{Input:} $\hat{x}_0$, $\hat{Q}_0$, $\hat{R}_0$, $\phat_0$, $\Psi_0$, $m$, $P$, $y_i,\, i=1,2,\ldots$
\State \textbf{Output:} $\hat{Q}_k$, $\hat{R}_k$, $\phat_k$, $\hat{x}_k$
\For{$k = 1 \to n$}
\If{$i>m+P$}
	\State Calculate $D_k$ and $b_k$ \Comment{Eq.~\eqref{eq:lincoveqn}}
	\State Calculate the Riemannian Gradient\Comment{Eq.~\eqref{eq:grad}}
	\State Calculate the Riemannian Hessian\Comment{Eq.~\eqref{eq:hess}}
	\State Use Algorithm~\ref{alg:rtr} to obtain $\hat{Q}_i$ and $\hat{R}_i$
	\State Update $\hat{x}_k$ and $\phat_k$ \Comment{Eq.~\eqref{eq:akf}}
\EndIf
\EndFor
\end{algorithmic}
\end{algorithm}

\section{Stability Analysis}\label{sec:stab}
In this section, the main contributions of this paper, i.e., stability of the RTR-based covariance estimation scheme
and the adaptive Kalman filter using the RTR-based covariance estimates is
presented.

\subsection{Convergence of the noise covariance estimates}

The RTR method, by design, ensures that $\hat{Q}_k$ and $\hat{R}_k$ are SPD.
Starting from SPD initial guesses the following results establish the
convergence of the RTR-based noise covariance estimators by comparing them to
the RLS solution.
\begin{prop}\label{prop:alwaysSPD}
	Given that $D_k$ is persistently excited, $\Pr\{\hat{Q}^{RLS}_i\in\Pn, \hat{R}^{RLS}_k \in\Pp,\, \forall i > k\} \ktendsto 1$.
\end{prop}
\begin{proof}
	The convergence of the batch least squares estimate $\thhat_k$ in the mean
	squared sense to the true value $\theta^*$ was established given that the
	combined coefficient matrix $D = [D_1^T, D_2^T,\ldots]^T$ is full column
	rank~\cite[Theorem 8]{dunik2018design}.  Since, $D_k$ is persistently
	excited, the full rank condition is automatically satisfied. Since,
	convergence in the mean squared sense implies convergence in probability,
	we know that $\Pr\{\|\thhat_k - \theta^*\| > 0\} \to 0$. Consequently, for
	any constant $\delta>0$, $\Pr\{\|\thhat_k - \theta^*\| < \delta\} \ktendsto
	1$.  We know that the true $Q$ and $R$ which are formed from the $\theta^*$
	elements are SPD.  Hence, there exists a $\delta$ such that $\forall\,
	\thhat_k\, : \|\thhat_k - \theta^*\| < \delta$, $\thhat_k$ is such that the
	matrices $\hat{Q}_k$ and $\hat{R}_k$ formed by its elements are SPD.
	Picking such a $\delta$ ensures that $\Pr\{\hat{Q}^{RLS}_k\in\Pn,\,
	\hat{R}^{RLS}_k\in\Pp\} \ktendsto 1$ which in turn ensures the following
	statement.  \[\Pr\{\exists i > k,\, \hat{Q}^{RLS}_i\notin\Pn,
	\hat{R}^{RLS}_i\notin\Pp\} \ktendsto 0.\] The negation of the above
	statement proves the statement of the proposition.
\end{proof}

\begin{prop}\label{prop:uniqueSolution}

	Given $\hat{Q}_k\in\Pn$, $\hat{R}_k\in\Pp$,
	$\Psi_k\in\mathbb{P}_{m_q+m_r}$, and the one step RLS and RTR solutions
	denoted by $(\hat{Q}^{RLS}_{k+1}, \hat{R}^{RLS}_{k+1})$ and
	$(\hat{Q}^{RTR}_{k+1}, \hat{R}^{RTR}_{k+1})$ respectively, if
	$\hat{Q}^{RLS}_{k+1}\in\Pn$ and $\hat{R}^{RLS}_{k+1}\in\Pp$ then
	$(\hat{Q}^{RLS}_{k+1}, \hat{R}^{RLS}_{k+1}) = (\hat{Q}^{RTR}_{k+1},
	\hat{R}^{RTR}_{k+1})$.

\end{prop}
\begin{proof}
	The cost function given in Eq.~\eqref{eq:lscost} is quadratic with a
	positive definite Euclidean Hessian and is hence convex in the argument
	$\theta$. As a result, the recursive least squares minimizer produces
	unique solutions up to the error due to the stopping criterion. Similarly,
	the choice of constants in Algorithm~\ref{alg:rtr} and the usage of exact
	Hessian ensures that $\lim\limits_{k\to\infty} \grad\, J_k = 0$ for the RTR
	algorithm~\cite[Theorem 7.4.4]{absil2009optimization}. Since,
	$\hat{Q}_k\in\Pn$ and $\hat{R}_k\in\Pp$, $\grad\, J_k = 0 \implies
	(\nabla_Q \Jbar_k, \nabla_R \Jbar_k) = (0,0)$. Hence, this solution exactly
	matches the solution from the RLS step up to the error induced due by the
	stopping criterion. The statement of the proposition follows.
\end{proof}

\begin{thm}\label{thm:cov}

	Given that $D_k$ is persistently excited, the sequences $\hat{Q}^{RTR}_k$ and
	$\hat{R}^{RTR}_k$ found using the Algorithm~\ref{alg:rtr} converge to their
	true values, $Q$ and $R$ respectively, in probability.

\end{thm}
\begin{proof}
	From Proposition~\ref{prop:alwaysSPD}, we know that
	$\Pr\{\hat{Q}^{RLS}_i\in\Pn,\, \hat{R}^{RLS}_i\in\Pp,\, \forall i > k\} \ktendsto 1$.
	Hence, from Proposition~\ref{prop:uniqueSolution}, we get that
	$\Pr\{(\hat{Q}^{RLS}_k, \hat{R}^{RLS}_k) \neq (\hat{Q}^{RTR}_k,
	\hat{R}^{RTR}_k)\} \ktendsto 0$. Since the RLS solution
	converges to the true value in the mean squares sense and the $RTR$
	solution matches the least squares solution with probability $1$ as
	$k\to\infty$, the RTR solution converges in probability to $(Q,R)$.
\end{proof}

\subsection{Convergence of the state error covariance matrix}
The stability properties of the adaptive Kalman filter using the RTR-based noise covariance
estimates are established through the following statement.
\begin{prop}\label{prop:expstab}
	Given the noise covariance estimates $\hat{Q}_k$ and $\hat{R}_k$ from the
	RTR Algorithm described in Algorithm~\ref{alg:rtr}, the adaptive Kalman
	filter from Eq.~\eqref{eq:akf} is exponentially stable.
\end{prop}
\begin{proof}
	The proof follows from~\cite[Theorem 5.3]{anderson1981detectability} and
	Definitions~\ref{defn:obs} and~\ref{defn:con}. The observability Gramian for
	the pair $(F_k, \hat{R}^{-\half}_k H_k)$ corresponding to the adaptive
	Kalman filter is given by
	\[
		\hat{M}_{k+m,k} = \sum\limits_{i=k}^{k+m} \phi_{i,k}^T H_i^T \hat{R}_i^{-1} H_i \phi_{i,k}
	\]
	Since, $\hat{R}_k \succ 0$ and the observability Gramian $M_{k+m,k}$ for
	the known case is SPD, the observability Gramian $\hat{M}_{k+m,k}$ for the
	adaptive Kalman filter using the RTR noise covariance matrix estimates is
	also SPD. Hence, the pair $(F_k, \hat{R}^{-\half}_k H_k)$ is
	uniformly observable. Given that $\hat{Y}_{k+s,k}\succ\bm{0}$, the
	controllability Gramian for the pair $(F_k, \hat{Q}_k^{\half})$
	corresponding to the adaptive Kalman filter for the same $s>0$ is given by
	\[
		\hat{Y}_{k+s,k} = \sum\limits_{i=k}^{k+s} \phi_{i,k} \hat{Q}_i \phi_{i,k}^T
	\]
	Since $\hat{Q}_k\succ 0$, the $\hat{Y}_{k+s,k}$ is non-singular and
	the pair $(F_k, \hat{Q}_k^{\half})$ is uniformly controllable. Hence
	from~\cite[Theorem 5.3]{anderson1981detectability}, the adaptive Kalman
	filter is exponentially stable.
\end{proof}

The following statement establishes the convergence of the state error
covariance matrix of the adaptive Kalman filter.
\begin{thm}
	The state error covariance matrix sequence $\phat_k$ of the adaptive Kalman
	filter converges to the state error covariance matrix $\popt_k$ of the
	optimal Kalman filter with probability $1$.
\end{thm}
\begin{proof}
	Consider three covariance sequences $\phat_k$, $P_k$ and $\popt_k$ given by~\cite{moghe2019adaptive}
	\begin{align}
		\phat_{k+1} &= \hat{\bar{F}}_k \hat{P}_k \hat{\bar{F}}_k^T + \khat_k \hat{R}_k \khat_k^T + \hat{Q}_k \\
		P_{k+1} &= \hat{\bar{F}}_k P_k \hat{\bar{F}}_k^T + \khat_k R \khat_k^T + Q \\
		\popt_{k+1} &= \bar{F}_k \popt_k \bar{F}_k^T + K_k R K_k^T + Q
	\end{align}
	wherein, $\bar{F}_k = F_k - \khat_k H_k$, $\hat{\bar{F}}_k = F_k - \khat_k H_k$,
	\begin{align*}
		\hat{K}_k = F_k \hat{P}_k H_k^T(H_k \hat{P}_k H_k^T + \hat{R}_k)^{-1} \\
		K_k = F_k \popt_k H_k^T(H_k \popt_k H_k^T + R)^{-1}
	\end{align*} 
	Each of the three sequences denotes the one-step predictor state covariance
	matrix. The sequence $\phat_k$ denotes the apparent state error covariance
	matrix of the adaptive Kalman filter and uses the noise covariance matrix
	estimates for its propagation. The sequence $P_k$ denotes the actual state
	error covariance matrix of the adaptive Kalman filter and uses the Kalman
	gain from the apparent covariance sequence along with the true noise covariance
	matrices. The sequence $\popt_k$ denotes the optimal state error covariance
	which represents the case when $Q$ and $R$ are fully known. We will first
	prove the equivalence of $\phat_k$ and $P_k$ in the limit. Assuming the
	same error covariance at the initial time, the sequence formed by
	differencing $\phat_k$ and $P_k$ is given by
	\begin{align*}
		\phat_{k+1} - P_{k+1} = (F_K - \khat_k H_k) (\phat_k - P_k) (F_K - \khat_k H_k)^T + \khat_k (\hat{R}_k - R) \khat_k^T + (\hat{Q}_k - Q)
	\end{align*}
	Since the RTR method ensures that $\hat{Q}_k$ and $\hat{R}_k$ are SPD, we conclude that $F_k - \khat_k H_k$ is
	exponentially stable from Proposition~\ref{prop:expstab}. Additionally, from Theorem~\ref{thm:cov}, both
	$\hat{Q}_k$ and $\hat{R}_k$ converge in probability to $Q$ and $R$
	respectively as $k\to\infty$. Hence, the exponentially stable matrix sequence converges to
	zero in probability, i.e., $\phat_k \ptendsto P_k$. From the expression for
	$\khat_k$, since $\phat_k \ptendsto P_k$ and $\hat{R}_k \ptendsto R$, we
	get
	\[
	\khat_k \ptendsto F_k P_k H_k^T(H_k P_k H_k^T + R)^{-1}
	\]
	as $k\to\infty$. Hence, the matrix sequence for $P_k$ and $\popt_k$ is
	identical in the limit as $k\to\infty$ with probability $1$. Invoking
	Proposition~\ref{prop:expstab}, the state transition matrix $F_k - K_k H_k$
	is exponentially stable. The matrix sequence $\popt_k$ has a unique
	limit~\cite[Theorem 7.5]{jazwinski2007stochastic}. Hence, $\phat_k\ptendsto
	\popt_k$ as $k\to\infty$. Finally, we conclude that $\phat_k \ptendsto
	\popt_k$ as $k\to\infty$.

\end{proof}

\section{Numerical Simulations}\label{sec:sim}
A LTV system is simulated in this section to demonstrate the RTR based adaptive
Kalman filter. The dynamics of the system given in Eq.~\eqref{eq:sys} with the
system matrices given below~\cite{dunik2018design}.
\begin{align}
	F_k = \begin{bmatrix} 0 & 1 \\ -a_k b_k & -(a_k + b_k) \end{bmatrix}\quad G_k = \begin{bmatrix} 1 \\ 1 \end{bmatrix}\quad H_k = \begin{bmatrix} 1 & d_k \end{bmatrix}
\end{align}
wherein, $\{a_k,b_k\}=c_k \pm i(0.4+0.2\sin(2\pi k/\tau)$, $c_k =
-0.7+0.2\cos(2\pi k/\tau)$, $d_k = 2sin(10\pi k/\tau)$, and $i$ is the
imaginary unit. The measurements are assumed to be available every $1/\tau$
second with $\tau = 10^4$. The true noise covariance matrices are given by
$Q=\begin{bmatrix} 3 & 1 \\ 1 & 2 \end{bmatrix}$ and $R = 2$. For the
purposes of the simulation, the control inputs $u_k$ are assumed to be
drawn from a unit normal distribution. The number of measurements stacked
at every time step are $m=3$. Fig.~\ref{fig:qnorm} shows the Frobenius norm
of the error in estimating the $Q$ matrix with the RTR and the RLS method.
The estimates from both methods are shown to converge to zero. A
similar trend is seen in the estimation error Frobenius norms in estimating
$R$ in Fig.~\ref{fig:rnorm} and the error between the state covariances of
the adaptive Kalman filter and the optimal Kalman filter with known noise
covariance matrices shown in Fig.~\ref{fig:pnorm}. The difference between
the two methods is seen when comparing the transient $\hat{Q}_k$
eigenvalues shown in Fig.~\ref{fig:lambdaQ}. The RLS method sometimes leads
to a negative eigenvalue while RTR method lower bounds the eigenvalue by a
prescribed minimum value of $0.1$ (Remark~\ref{rm:lmin}). Since the $R$
matrix is a scalar, A similar trend is seen in Fig.~\ref{fig:lambdaR} that
shows the time history of its estimate $\hat{R}_k$. The result of negative
eigenvalues of the noise covariance matrix estimates culminates as an
inconsistent non positive definite error state covariance $\phat_{k+1|k}$
shown in Fig.~\ref{fig:lambdaP}.

\begin{figure}[!htpb]
	\centering
	\includegraphics{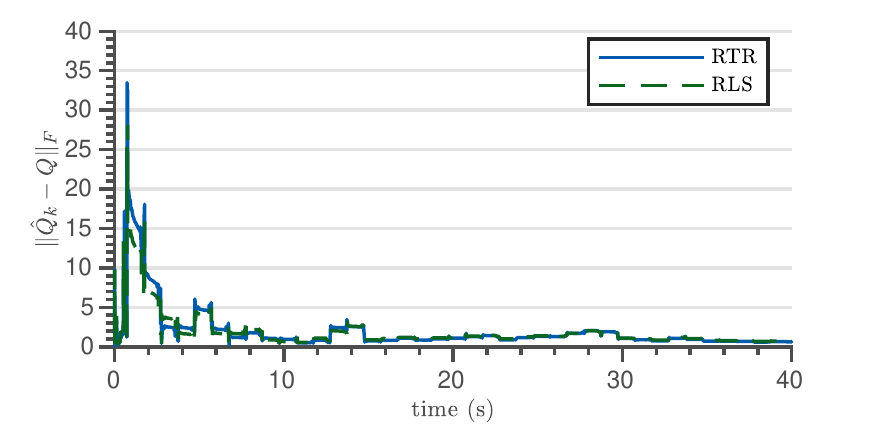}
	\caption{The Frobenius norm of the $Q$ estimation error.}
	\label{fig:qnorm}
\end{figure}

\begin{figure}[!htpb]
	\centering
	\includegraphics{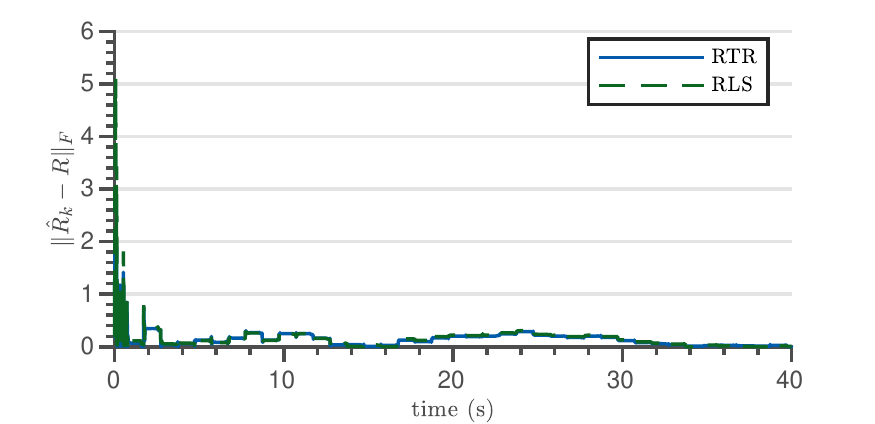}
	\caption{The Frobenius norm of the $R$ estimation error.}
	\label{fig:rnorm}
\end{figure}

\begin{figure}[!htpb]
	\centering
	\includegraphics{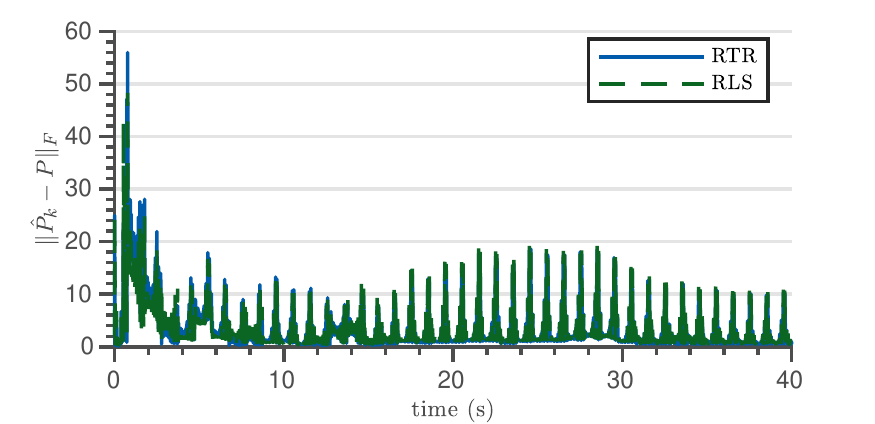}
	\caption{The Frobenius norm of the estimation error in the estimated state
	error covariance matrix $\phat_{k|k=1}$ and the optimal $\popt_{k|k-1}$.}
	\label{fig:pnorm}
\end{figure}

\begin{figure}[!htpb]
	\centering
	\includegraphics{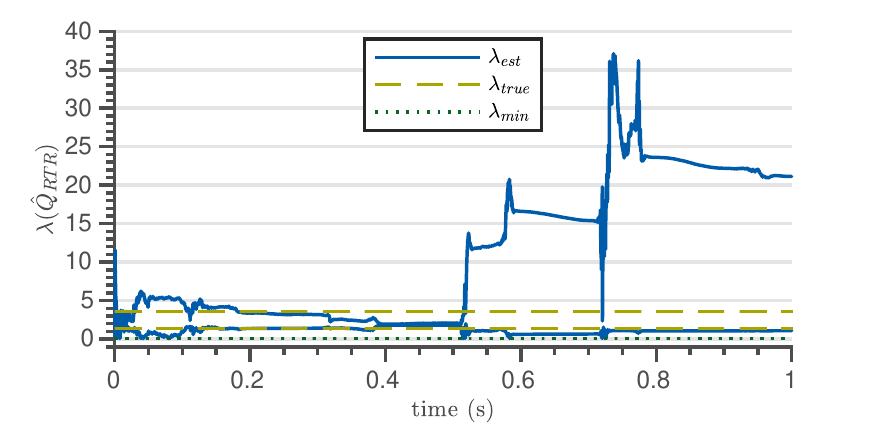}
	\includegraphics{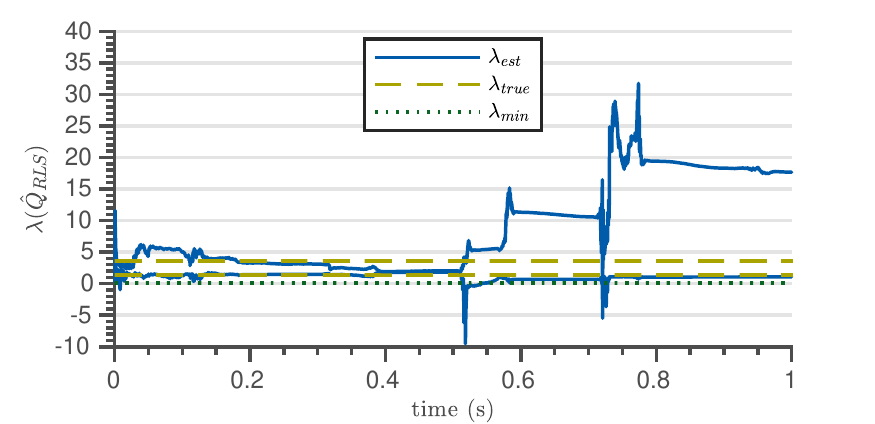}
	\caption{The transient eigenvalues of $\hat{Q}_k$, the true eigenvalues of $Q$, and $\lambda_{min}=0.1$ from Remark~\ref{rm:lmin}.}
	\label{fig:lambdaQ}
\end{figure}

\begin{figure}[!htpb]
	\centering
	\includegraphics{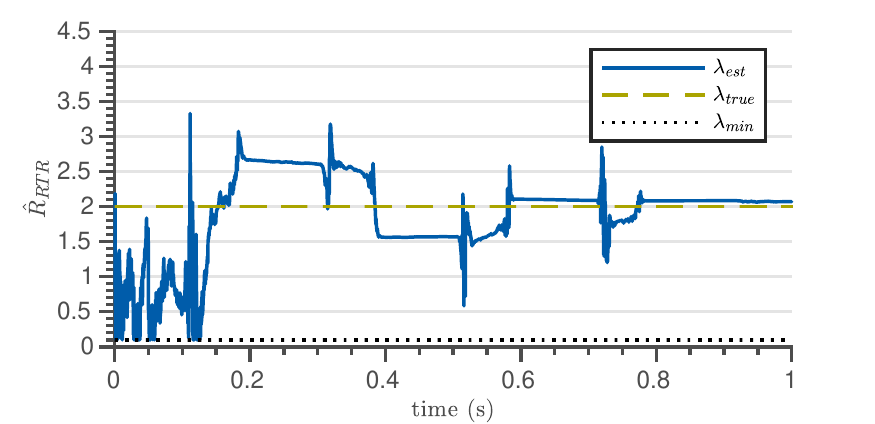}
	\includegraphics{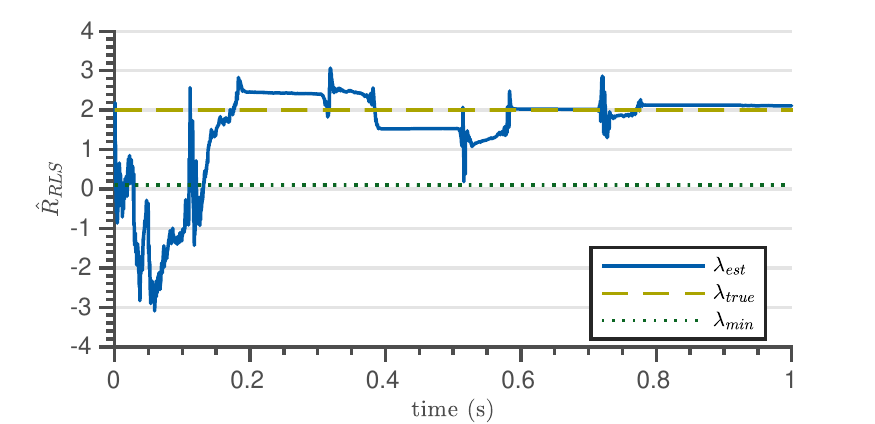}
	\caption{The transient values of $\hat{R}_k$, the true $R$, and $\lambda_{min}=0.1$ from Remark~\ref{rm:lmin}.}
	\label{fig:lambdaR}
\end{figure}

\begin{figure}[!htpb]
	\centering
	\includegraphics{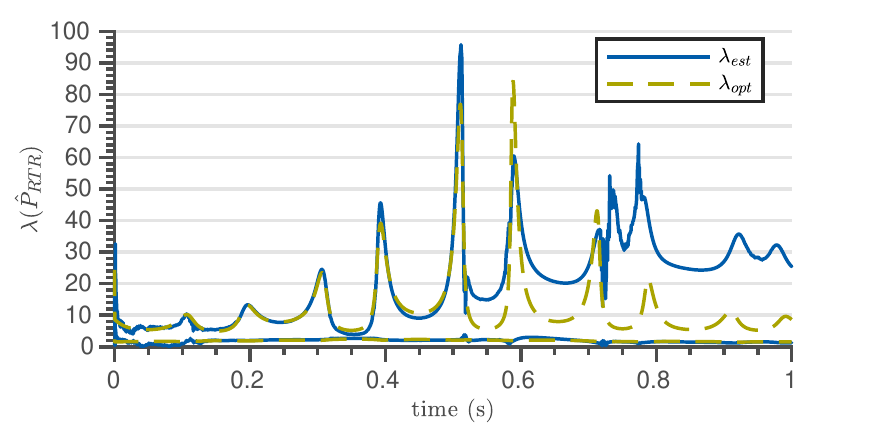}
	\includegraphics{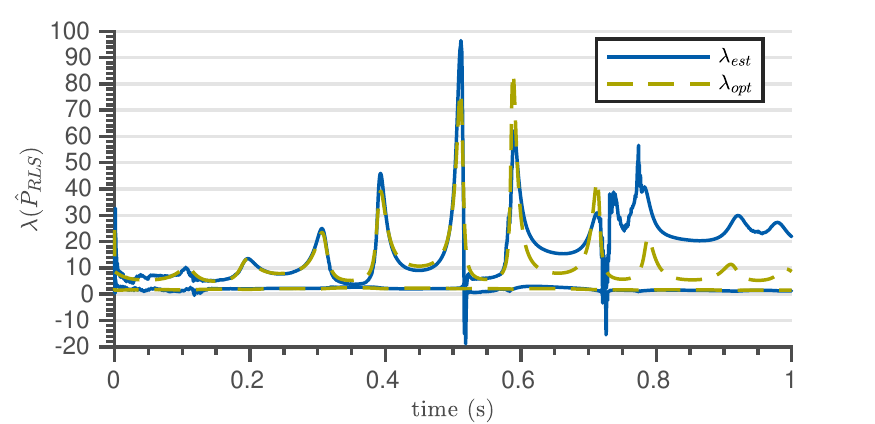}
	\caption{The transient eigenvalues of predicted state error covariance for the RTR-based adaptive Kalman filter.}
	\label{fig:lambdaP}
\end{figure}

\section{Conclusion}\label{sec:conc}

A Riemannian trust-region (RTR) based adaptive Kalman filter to estimate the
states as well as the process and measurement noise covariance matrices (CM)
for an discrete observable linear time-varying system is presented in this
paper. Rigorous convergence guarantees are provided for the noise CM estimate
as well as the state error CM of the adaptive Kalman filter formulated using
the noise CM estimates. To ensure convergence, the only assumptions required in
this formulation are uniform observability, uniform controllability, and
sufficient excitation of the system matrices. In fact, if the system matrices
are not sufficiently excited, the adaptive Kalman filter is still exponentially
stable by virtue of positive definite noise CM estimates.  The results provided
in this paper can be extended to handle detectable systems by using a state
transformation to reveal and ignore the unobservable subspace of the system.
The convergence of the RTR optimization method to the optimal value is
theoretically fast and confirmed through numerical simulations because of
the quadratic nature of the cost function. However, the RTR method is an
iterative procedure and may be computationally expensive for problems of
larger dimensionality. Formulating non-iterative schemes to ensure positive
definite estimates presents a potential direction of future research.

\bibliographystyle{unsrt}  
\bibliography{root}  

\end{document}